\documentclass{article}
\usepackage[utf8]{inputenc}
\usepackage{amsmath}
\usepackage{nomencl}
\makenomenclature
\usepackage{multicol}
\usepackage{amsmath}
\usepackage{amsthm}
\theoremstyle{definition}
\newtheorem{definition}{Definition}
\theoremstyle{plain}
\newtheorem{theorem}{Theorem}
\newtheorem{lemma}{Lemma}
\usepackage{cite}
\usepackage{hyperref}
\hypersetup{
    colorlinks=true,
    linkcolor=blue,
    filecolor=magenta,      
    urlcolor=cyan,
}
\usepackage{url}
\usepackage{dirtytalk}
\newcommand{\cmmnt}[1]{}
\usepackage[a4paper, total={6in, 8in}]{geometry}

\usepackage{amssymb, graphics, setspace}

\usepackage{xcolor}

\newcommand{\mathsym}[1]{{}}
\newcommand{\unicode}[1]{{}}

\title{The Lagrangian remainder of Taylor's series, distinguishes $\mathcal{O}(f(x))$ time complexities to polynomials or not}
\author{Nikolaos P. Bakas\thanks{Computation-based Science and Technology Research Center, The Cyprus Institute, 20 Konstantinou Kavafi Street, 2121, Aglantzia Nicosia, Cyprus. e-mail: n.bakas@cyi.ac.cy.}, \hspace{3mm} Elias Kosmatopoulos\thanks{Democritus University of Thrace, Electrical and Computer engineering. e-mail: kosmatop@ee.duth.gr.}, Mihalis Nicolaou, \thanks{Computation-based Science and Technology Research Center, The Cyprus Institute, 20 Konstantinou Kavafi Street, 2121, Aglantzia Nicosia, Cyprus. e-mail: 	m.nicolaou@cyi.ac.cy}, \hspace{3mm} Savvas A. Chatzichristofis\thanks{Intelligent Systems Lab \&  Department of Computer Science, Neapolis University Pafos, 2 Danais Avenue, 8042 Pafos, Cyprus. e-mail: s.chatzichristofis@nup.ac.cy.} }

\date{}

\begin{document}

\maketitle

\section*{Abstract}

The purpose of this letter is to investigate the time complexity consequences of the truncated Taylor series, known as Taylor Polynomials \cite{bakas2019taylor,Katsoprinakis2011,Nestoridis2011}. In particular, it is demonstrated that the examination of the $\mathbf{P=NP}$ equality, is associated with the determination of whether the $n^{th}$ derivative of a particular solution is bounded or not. Accordingly, in some cases, this is not true, and hence in general.

\section{Univariate complexity}

\begin{definition}
Let the given problem is a known analytic function $f$ of one variable $x \in \mathbf{Z}^{+}$. Initially, the authors consider one-dimensional $x$, and later they generalize the results. Respectively, the \textit{time complexity} of the given problem, according to the literature \cite{sipser2006introduction}, may be written in the generic form of: 
\[\mathcal{O}(f(x)).\]
\end{definition}

The \textit{execution time }is usually calculated by some elementary algebraic operations of integers or real numbers \cite{johnson1979computers}, thus,  this assumption is considered to be adequate and valid. Accordingly, the Taylor series expansion of $f$  at $x+x_0$ may be written by
\[f(x)=f({{x}_{0}})+\frac{{f}'({{x}_{0}})}{1!}(x-{{x}_{0}})+\frac{{f}''({{x}_{0}})}{2!}{{(x-{{x}_{0}})}^{2}}+\cdots +\frac{{{f}^{(n)}}({{x}_{0}})}{n!}{{(x-{{x}_{0}})}^{n}}+\cdots,\]
with infinite terms, that is 
\begin{equation}
    f(x)={{T}_{n}}(x)+{{R}_{n}}(x),
    \label{eq:t-r}
\end{equation}
where ${{T}_{n}}(x)$ is the Taylor polynomial of order $n$, and ${{R}_{n}}(x)$ the \underline{remainder} of the $n^{th}$ degree Taylor polynomial.

If the function $f$ is analytic, it has $n+1$ derivatives for each point in the interval $\left| x-{{x}_{0}} \right|\le r$. Moreover, if the derivatives for each point are in the interval $\mid {{f}^{\left( n+1 \right)}}\left( x \right)\mid \le M$, then, the Lagrangian form \cite{kline1998calculus,apostol1967calculus} of the remainder may be expressed by:
\begin{equation}
    |{{R}_{n}}(x)|\le \frac{M}{(n+1)!}|x-{{x}_{0}}{{|}^{n+1}}\forall x:|x-{{x}_{0}}|\le r.
    \label{eq:reminder}
\end{equation}

This is subsequent of the explicit form of the remainder, stating that there exist a $\xi$ among $x$, and $x_0$, such that 

\[
 {{R}_{n}}(x)\le \frac{f^{(n)}(\xi)}{(n+1)!}(x-{{x}_{0}}{{)}^{n+1}}.
\]

\begin{theorem}
If $\mid {{f}^{\left( n+1 \right)}}\left( x \right)\mid \le M$, the algorithm with $\mathcal{O}(f(x))$ complexity, runs in polynomial time. 
\end{theorem}
\begin{proof}
Equations \ref{eq:t-r} and \ref{eq:reminder} obtain that:
\[f(x)={{T}_{n}}(x)+{{P}_{n+1}}(x),\]
thus, 
\[\mathcal{O}(f(x))=x^{n+1}.\]
which apparently is a polynomial, and hence of class $\mathbf P$.
\end{proof}

\begin{lemma}
If $\mid {{f}^{\left( n+1 \right)}}\left( x \right)\mid > M$, 
$f$ cannot be expressed as polynomials.
\end{lemma}
\begin{proof}
\sloppy By utilizing Borel's theorem \cite{narasimhan1985analysis}, stating that any formal series $\sum _{n=0}^{\infty }a_{n}\left(x-x_0\right)^{n}$ is the Taylor series of a $C^\infty$-smooth function defined in an open neighborhood of $x_0$, it is  derived that if $f^{(n)}$ is not bounded, $f$ cannot be written as a power series, and hence as a polynomial, thus the problem is not in $\mathbf P$. In other words, if the problem was in $\mathbf P$, it could be written as a polynomial, and this would be the Taylor series, which is absurd as no $n$ exist such that the $f^{(n)}$ is limited by a $M$.
\end{proof}

\section{Multiple dimensions}
\sloppy
The given problem might have execution time depending on two or more variables $\mathbf{x}=\{x_1,x_2,\dots,x_m\} \in \mathbf{Z}^{m+}$, hence is is a function $f$ of a variable $\mathbf x$, that is 
\[\mathcal{O}(f(\mathbf{x})).\]

The above mentioned procedure may be applied to define if the complexity of the algorithm is polynomial, for each dimension $i \in \{1,2,\dots,m\}$, concerning the others as constants. If for one dimension $\mathcal{O}(f({x_i}))$ is not a polynomial, then $\mathcal{O}(f(\mathbf{x}))$ is not.

\section{Examples}

The presented analysis of computational time with Taylor series constitutes a basis for the investigation of whether a given complexity of $\mathcal{O}(\cdot)$ is a polynomial or not.

\subsection{Example of $\left|f^{(n)}(x)\right|> M$}
\subsubsection{$f(x)=e^x$}

In this simple case, one can easily observe that $f^{(n)}(x)=e^x \forall n$, which is not bounded by any $M$ for all $x$.

\[
r=\lim _{n\rightarrow \infty }\left|{\frac {a_{n}}{a_{n+1}}}\right|=\lim_{n\to\infty} \frac{\frac{x^n}{n!}}{\frac{x^{n+1}}{{n+1}!}}=\lim_{n\to\infty} \frac{n+1}{x}=\infty,
\]

and the series converges everywhere.

\subsubsection{$f(x)=2^x$}

A variety of algorithms execute in times calculated by a number, raised to the $x^{th}$ power, for example the matrix chain multiplication via brute-force search is $2^x$, a well as the lower bounds for the $\mathbf{AC}^0$ problem \cite{impagliazzo2001problems}. In this case:
\[\frac{df}{dx} = 2^x log(2),\]
\[\frac{d^2f}{dx^2} = 2^x log^2(2),\] and hence
\[\frac{d^nf}{dx^n} = 2^x log^n(2),\]

which is not bounded by any $M$ for all $x$.

 Apparently,
\[{{2}^{x}}=1+\frac{x\log (2)}{1!}+\frac{{{x}^{2}}{{\left( \log (2) \right)}^{2}}}{2!}+\cdots +\frac{{{x}^{n}}{{\left( \log (2) \right)}^{n}}}{n!}\]

\[r=\underset{n\to \infty }{\mathop{\lim \inf }}\,\frac{1}{\sqrt[n]{|{{a}_{n}}|}}=\underset{n\to \infty }{\mathop{\lim \inf }}\,\frac{\sqrt[n]{n!}}{x\left( \log (2) \right)}=\infty,\]

or equivalently,

\[r=\underset{n\to \infty }{\mathop{\lim }}\,\left| \frac{{{a}_{n}}}{{{a}_{n+1}}} \right|=\frac{\frac{{{x}^{n}}{{\left( \log (2) \right)}^{n}}}{n!}}{\frac{{{x}^{n+1}}{{\left( \log (2) \right)}^{n+1}}}{(n+1)!}}=\lim_{n\to\infty} \frac{(n+1)}{x\log(2)}=\infty,\]

and the series converges everywhere.

\subsubsection{$f(x)=2^{x^{1/2}}$}
\label{sec:sqrt2}

E.g. the monotone circuits computing \cite{alon1987monotone}. This is an example which though seems exponential, it is not verified if it indeed in not a polynomial (in contrast with \ref{sec:sqrt4}).

\[
\frac{d}{dx}2^{x^{1/2}} = \frac{2^{x^{1/2} - 1} log(2)}{x^{1/2}},
\]

and hence, 

\[
\underset{x\to \infty }{\lim} \frac{d}{dx}2^{x^{1/2}} = \infty.
\]






\subsubsection{$f(x)=x^{\log(x)}$}

Any monotone formula for the computation of a monotone function in NP must have size at least $\Omega(x^{\log(x)})$ \cite{razborov1990applications,robere2016exponential}.

\[
\frac{d}{dx}{x^{\log_2(x)}} = \frac{2 x^{\frac{log(x)}{log(2)} - 1} log(x)}{log(2)},
\]

where $log$ is the natural logarithm, and hence, 

\[
\underset{x\to \infty }{\lim} \frac{d}{dx}x^{\log(x)} = \infty.
\]

Accordingly,


\[\frac{\partial ^nx^{\log _2(x)}}{\partial x^n}=(-1)^n x^{\log _2(x)-n} \left(-\log _2(x)\right)_n\text{/;}n\in \mathbb{Z}\land
n\geq 0\land x \log (2)\neq \log (x)\]

and $(\cdot)_n$ is the Pochhammer Symbol, with 

\[(\xi)_n=(\Gamma(\xi+n))/(\Gamma(\xi))=(\xi+1) \dots (\xi+n-1).\]

Thus $\left(-\log _2(x)\right)_n$, never vanishes for $x>1$. Additionally, 

\[\lim_{x\to \infty } \, x^{\log (x)/\log (2)-n}=\infty,\]

an hence the $n^{th}$ derivative of $f$ is not bounded. The alliterating sign $(-1)^n$ results in $+/- \infty$ for the limit of $\frac{\partial ^nx^{\log _2(x)}}{\partial x^n}$. 

\subsection{Examples of $\left|f^{(n)}(x)\right| \leq M$}

\subsubsection{$f(x)=xlog(x)$}
This is a commonly resultant time (e.g.  comparison sort). 

\[\frac{df}{dx}=log(x)+x\frac{1}{x},\]
\[\frac{d^2f}{dx^2}=\frac{1}{x},\]

which is $ \leq M, \forall x>1$, even for the second derivative ($n=2$). 

In this case, $x \in \mathbf{Z}^+$, and hence, the radius of convergence is required to be $>0$. For $x>\frac{1}{2}$:

\[
log(x)=\sum_{n=1}^{\infty}\frac{\left(\frac{x-1}{x}\right)^n}{n},
\]

and hence 

\[
x log(x)=\sum_{n=1}^{\infty}x\frac{\left(\frac{x-1}{x}\right)^n}{n},
\]

thus, it is obtained that:

\[
r=\lim_{n\to\infty} |\frac{a_n}{a_{n+1}}|=\lim_{n\to\infty} \frac{x\frac{\left(\frac{x-1}{x}\right)^n}{n}}{x\frac{\left(\frac{x-1}{x}\right)^{n+1}}{n+1}}=\lim_{n\to\infty}\frac{n+1}{n\frac{x-1}{x}}=\frac{x}{x-1}>1 \forall x \in \mathbf{Z}^+ \footnote{if the radius of convergence $r>1$, is enough criterion for generalization if the Taylor polynomials are written about a point $x_0$, instead of zero, as one may write the $\mathcal{O}(f(x))$ with polynomials in branches $\forall x_0$.}.
\]

\subsubsection{$f(x)=2^{\log _2(x)}$}
\label{sec:sqrt4}
\[\frac{d2^{\log _2(x)}}{dx}=1.\]

\subsubsection{$f(x)=2^{\log _2\left(\log _2(x)\right)}$}

Results to:

\[\frac{d2^{\log _2\left(\log _2(x)\right)}}{dx}=\frac{1}{x \log (2)}, f^{(2)}(x)=-\frac{1}{x^2 \log (2)},\]

\[f^{(3)}(x)=\frac{2}{x^3 \log (2)}, f^{(4)}(x)=-\frac{6}{x^4 \log (2)}, f^{(5)}(x)=\frac{24}{x^5 \log (2)}, \]

thus

\[\frac{d^n2^{\log _2\left(\log _2(x)\right)}}{dx^n}=(-1)^{n-1}\frac{(n-1)!}{x^n \log (2)}\]

and hence, 

\[\lim_{x\to\infty} \sup f^{(n)}(x)=\lim_{x\to\infty} \inf f^{(n)}(x)=0; \forall n\ge1.\]


\section{Conclusion}

P versus NP problem is to determine whether every problem whose solution can be verified in polynomial time, can be also solved in polynomial time \cite{cook2006p}. In this letter, the authors highlighted that for a given complexity of $\mathcal{O}(f(x))$ for the solution of a problem, under certain criteria for a $n^{th}$ derivative of $f$, $f^{(n)}(x)$, this problem cannot be considered as a polynomial. We do not consider non-analytic, non differentiable problems, which anyway do not affect the obtain conclusion. Furthermore, some specific $\mathbf{NP}$ problems, have known exponential \underline{lower bounds}, e.g.  \cite{ukkonen1983exponential,kikot2012exponential,alon1987monotone,razborov1990applications,robere2016exponential}. In some of these cases, the $n^{th}$ derivative of $f$ is not bounded. Hence, in such cases, their solution can be \say{quickly} verified, in polynomial time, but on the other hand, by utilizing the aforementioned rationale, the lower bound of their solution cannot be expressed with a polynomial function of computational time, and hence cannot belong to class $\mathbf{P}$. Accordingly, the authors derive that in these cases (and hence in general), the proposed rationale, apart from distinguishing polynomial or not algorithms, based on existing lower bounds, might answer the question if 

\vspace{5mm}

\[\mathbf P \ne \mathbf {NP}.\]

\vspace{13mm}
\nomenclature{${{T}_{n}}(x)$}{Taylor polynomial of order $n$}
\nomenclature{$n$}{Order of the truncated polynomial, as well as the highest utilized derivative.}
\nomenclature{$m$}{number of variables for multidimensional problems.}
\nomenclature{${{R}_{n}}(x)$}{remainder of the $n^{th}$ degree Taylor polynomial}
\nomenclature{${{P}_{n}}(x)$}{polynomial of order $n$}
\nomenclature{$x$}{integer $\in \mathbf{Z}^+$, indicating the order of the problem.}
\nomenclature{$r$}{radius of convergence.}
\begin{multicols}{2}
\printnomenclature[1cm]
\end{multicols}

\bibliographystyle{IEEEtran}
\bibliography{refs}

\begin{thebibliography}{10}
\providecommand{\url}[1]{#1}
\csname url@samestyle\endcsname
\providecommand{\newblock}{\relax}
\providecommand{\bibinfo}[2]{#2}
\providecommand{\BIBentrySTDinterwordspacing}{\spaceskip=0pt\relax}
\providecommand{\BIBentryALTinterwordstretchfactor}{4}
\providecommand{\BIBentryALTinterwordspacing}{\spaceskip=\fontdimen2\font plus
\BIBentryALTinterwordstretchfactor\fontdimen3\font minus
  \fontdimen4\font\relax}
\providecommand{\BIBforeignlanguage}[2]{{%
\expandafter\ifx\csname l@#1\endcsname\relax
\typeout{** WARNING: IEEEtran.bst: No hyphenation pattern has been}%
\typeout{** loaded for the language `#1'. Using the pattern for}%
\typeout{** the default language instead.}%
\else
\language=\csname l@#1\endcsname
\fi
#2}}
\providecommand{\BIBdecl}{\relax}
\BIBdecl

\bibitem{bakas2019taylor}
\BIBentryALTinterwordspacing
N.~P. Bakas, ``Taylor polynomials in high arithmetic precision as universal
  approximators,'' \emph{arXiv preprint arXiv:1909.13565}, 2019. [Online].
  Available: \url{https://arxiv.org/abs/1909.13565}
\BIBentrySTDinterwordspacing

\bibitem{Katsoprinakis2011}
E.~S. Katsoprinakis and V.~N. Nestoridis, ``{Partial sums of Taylor series on a
  circle},'' \emph{Annales de l'institut Fourier}, 2011.

\bibitem{Nestoridis2011}
V.~Nestoridis, ``{Universal Taylor series},'' \emph{Annales de l'institut
  Fourier}, 2011.

\bibitem{sipser2006introduction}
M.~Sipser \emph{et~al.}, \emph{Introduction to the Theory of
  Computation}.\hskip 1em plus 0.5em minus 0.4em\relax Thomson Course
  Technology Boston, 2006, vol.~2.

\bibitem{johnson1979computers}
D.~S. Johnson and M.~R. Garey, \emph{Computers and intractability: A guide to
  the theory of NP-completeness}.\hskip 1em plus 0.5em minus 0.4em\relax WH
  Freeman, 1979.

\bibitem{kline1998calculus}
M.~Kline, \emph{Calculus: an intuitive and physical approach}.\hskip 1em plus
  0.5em minus 0.4em\relax Courier Corporation, 1998.

\bibitem{apostol1967calculus}
T.~M. Apostol, ``Calculus. 1967,'' \emph{Jon Wiley \& Sons}, 1967.

\bibitem{narasimhan1985analysis}
R.~Narasimhan, \emph{Analysis on real and complex manifolds}.\hskip 1em plus
  0.5em minus 0.4em\relax Elsevier, 1985, vol.~35.

\bibitem{impagliazzo2001problems}
R.~Impagliazzo, R.~Paturi, and F.~Zane, ``Which problems have strongly
  exponential complexity?'' \emph{Journal of Computer and System Sciences},
  vol.~63, no.~4, pp. 512--530, 2001.

\bibitem{alon1987monotone}
N.~Alon and R.~B. Boppana, ``The monotone circuit complexity of boolean
  functions,'' \emph{Combinatorica}, vol.~7, no.~1, pp. 1--22, 1987.

\bibitem{razborov1990applications}
A.~A. Razborov, ``Applications of matrix methods to the theory of lower bounds
  in computational complexity,'' \emph{Combinatorica}, vol.~10, no.~1, pp.
  81--93, 1990.

\bibitem{robere2016exponential}
R.~Robere, T.~Pitassi, B.~Rossman, and S.~A. Cook, ``Exponential lower bounds
  for monotone span programs,'' in \emph{2016 IEEE 57th Annual Symposium on
  Foundations of Computer Science (FOCS)}.\hskip 1em plus 0.5em minus
  0.4em\relax IEEE, 2016, pp. 406--415.

\bibitem{cook2006p}
S.~Cook, ``The p versus np problem,'' \emph{The millennium prize problems}, pp.
  87--104, 2006.

\bibitem{ukkonen1983exponential}
E.~Ukkonen, ``Exponential lower bounds for some np-complete problems in a
  restricted linear decision tree model,'' \emph{BIT Numerical Mathematics},
  vol.~23, no.~2, pp. 181--192, 1983.

\bibitem{kikot2012exponential}
S.~Kikot, R.~Kontchakov, V.~Podolskii, and M.~Zakharyaschev, ``Exponential
  lower bounds and separation for query rewriting,'' in \emph{International
  Colloquium on Automata, Languages, and Programming}.\hskip 1em plus 0.5em
  minus 0.4em\relax Springer, 2012, pp. 263--274.

\end{thebibliography}

\end{document}